%% file: netys-cameraready-arxiv-new.tex
\documentclass[oribibl]{llncs}
\usepackage{lscape}
\usepackage{cite}
\usepackage{graphicx}
\usepackage{subfig}
\usepackage{amsmath}
\usepackage{algorithmic}
\usepackage{algorithm}
\usepackage{listings}
\usepackage{epstopdf}
\usepackage{amssymb}
\usepackage{times}
\usepackage{xspace}
\usepackage{soul} 
\usepackage{enumitem}
\usepackage{authblk}
\usepackage{url}
\usepackage{marginnote}
\usepackage{ulem} 
\usepackage{framed}
\usepackage{calc}
\usepackage{lipsum}
\usepackage[final]{changes}

\newtheorem{observation}[theorem]{Observation}

\newcommand{\CH}{{\ensuremath{\mathcal{C}}}\xspace}
\newcommand{\VRF}{{\ensuremath{\mathcal{A}}}\xspace}
\newcommand{\Cp}{{\ensuremath{\mathit{WP_{\CH}}}}\xspace}
\newcommand{\Ct}{{\ensuremath{\mathit{WC_\mathcal{T}}}}\xspace}
\newcommand{\SW}{{\ensuremath{\mathit{WB_\mathcal{Y}}}}\xspace}

\newcommand{\am}{\alpha_m\xspace}
\newcommand{\aw}{\alpha_w\xspace}

\newcommand{\typeA}{{\sc Linear}\xspace}
\newcommand{\typeB}{{\sc Exponential}\xspace}
\newcommand{\typeC}{{\sc Boinc}\xspace}
\newcommand{\typeD}{{\sc Boinc}\xspace}
\newcommand{\remove}[1]{}

\definechangesauthor[Evgenia]{ev}{magenta}

\title{Internet Computing: Using Reputation to Select Workers from a Pool\vspace{-1em}}  	
\author{Evgenia Christoforou\inst{1,2} \and Antonio Fern\'andez Anta\inst{1} \and
Chryssis Georgiou\inst{3} \and \newline Miguel A. Mosteiro\inst{4} }
\institute{
IMDEA Networks Institute, Madrid, Spain
\and 
Universidad Carlos III de Madrid, Madrid, Spain
\and
University of Cyprus, Nicosia, Cyprus
\and
Kean University, Union, NJ, USA\vspace{-1em} 
}   		     				           

\begin{document}                    

\maketitle          

\begin{abstract}
The assignment and execution of tasks over the Internet is an inexpensive solution in contrast with supercomputers. We consider an Internet-based Master-Worker task computing approach, such as 
SETI@home. A master process sends tasks, across the Internet, to worker processors. Workers execute, and report back a result. Unfortunately, the disadvantage of this approach is the unreliable nature of the worker processes. Through different studies, workers have been categorized as either malicious (always report an incorrect result), altruistic (always report a correct result), or rational (report whatever result maximizes their benefit). We develop a reputation-based mechanism that guarantees 
that, eventually, the master will always be receiving the correct task result.
We model the behavior of the rational workers through reinforcement learning, and we present three different reputation types to choose, for each computational round, the most reputable from a pool of workers. As workers are not always available, we enhance our reputation scheme to select the most responsive workers. We prove sufficient conditions for eventual correctness under the different reputation types. Our analysis is complemented by simulations exploring various scenarios. 
Our simulation results expose interesting trade-offs among the different reputation types, workers availability, and cost.\vspace{-1.2em} 
\end{abstract}

\vspace{-1.3em}
\input{intro}

\input{model-short}

\input{reputation}

\input{analysis}

\input{simulations-new-arxiv-new}

\noindent \paragraph{\bf Acknowledgments:}  
Supported in part by MINECO grant TEC2014- 55713-R, Regional Government of Madrid (CM) grant Cloud4BigData (S2013/ICE-2894, co- funded by FSE \& FEDER), NSF of China grant 61520106005, EC H2020 grants ReCred and NOTRE,  U. of Cyprus (ED-CG2015) and the MECD grant FPU2013-03792.

\vspace*{-4mm}


\newpage
\appendix
\input{appendix}

\end{document}

%% file: intro.tex

\section{Introduction}
\vspace*{-1em}
Internet-based computing has emerged as an inexpensive alternative for scientific high-performance computations. 
The most popular form of Internet-based computing is volunteer computing, where computing resources are volunteered
by the public to help solve (mainly) scientific problems. BOINC~\cite{boinc} is a popular platform where volunteer computing projects run, 
such as SETI@home~\cite{SETI}. Profit-seeking computation platforms, such as Amazon's Mechanical Turk~\cite{turk}, have also
become popular. One of the main challenges for further exploiting the promise of such platforms is the untrustworthiness of the
participating entities~\cite{boinc, volunteer, Kondoetal2007, Heienetal09}.

In this work we focus on Internet-based master-worker task computing, 
where a master process sends tasks, across the Internet, to worker processes to compute and return the result. Workers, however, might report incorrect results. Following~\cite{opodis2013, IEEETC14}, we consider three types of worker. Malicious\added[ev]{\footnote{We call these workers malicious for consistency with previous literature on Volunteer Computing~\cite{boinc}. This must not be confused with Byzantine malice assumed in classical distributed computing.}} workers that always report an incorrect result, altruistic workers that always report a correct result, and rational workers that report a result driven by their self-interest. In addition, a worker (regardless of its type) might be unavailable (e.g., be disconnected, be busy performing other tasks, etc). 
%
Our main contribution is a computing system where the master eventually obtains always the correct task result despite the above shortcomings. Our mechanism is novel in two fronts: 
(\emph{i}) it leverages the possibility of changing workers over time, given that the number of workers willing to participate is larger than the number of workers needed, and 
(\emph{ii}) it is resilient to some workers being unavailable from time to time.


Worker unreliability in master-worker computing has been studied from  
both a classical Distributing Computing approach and a Game Theoretic one. The first 
treats workers as malicious 
or altruistic.
Tasks
are redundantly allocated to different workers, and voting protocols that tolerate malicious workers have been designed 
(e.g., 
\cite{Sarmenta02, ALEX, PPL12}). The Game Theoretic approach views the workers as rational~\cite{Halp06,UDC,rational}, 
who follow 
the strategy that would maximize 
their benefit. In the latter approach, 
incentive-based mechanisms have been developed (e.g., \cite{CCS,NCA08}) that 
induce
workers to act correctly.

Other works  (e.g., \cite{opodis2013, IEEETC14}) have considered the co-existence of all three types of worker. 
In~\cite{IEEETC14}, 
a ``one-shot'' interaction between master and workers was implemented.
In that work, 
the master assigns tasks to workers 
without using 
knowledge of past interactions (e.g., on the behavior of the workers). 
In~\cite{opodis2013}, a mechanism was designed taking advantage of the repeated interaction (rounds) of the master with the workers. The
mechanism employs reinforcement learning~\cite{RLbook} both for  
the master and for the workers. In each round, the master
assigns a task to the same set of workers (which 
are assumed to be always available). The master may audit (with a cost) 
the responses of the workers and a reward-punishment scheme is employed. Depending on the answers, the master adjusts its probability
of auditing. 
Rational workers 
cheat (i.e., respond with an incorrect result to avoid the cost of computing) with some probability,
which over the rounds increases or decreases depending on the incentive received (reward or punishment). Rational workers
have an aspiration level~\cite{BM55} which determines whether a received payoff was satisfactory or not. To cope with malicious workers (whose behavior is not affected by the above mentioned learning scheme) a reputation scheme~\cite{survey07} was additionally employed.
The main objective is to ``quickly'' reach a round in the computation after which the master always receives the correct task result, 
with minimal auditing.


Unlike assumed in~\cite{opodis2013} (and most previous literature), in practice 
workers are not always available. For instance, Heien et al.~\cite{Heienetal09}
have found that in BOINC~\cite{boinc} only around 5\% of the workers are available more than 80\% of the time, and that half of the workers are available less than 40\% of the time. 
In this work, we extend the work in~\cite{opodis2013} to cope with worker unavailability. 

A feature that has not been leveraged in~\cite{opodis2013} and previous works is
the scale of Internet-based master-worker task computing systems. For example, in BOINC~\cite{boinc_main} 
active workers are around a few hundred thousand. 
In such a large system,
replicating the task and sending it to all workers is 
neither feasible nor practical.
On the other hand, randomly selecting a small number of workers to send the task does not guarantee 
correctness with minimum auditing.
For instance, consider a pool of workers where the malicious outnumber those needed for the computation. Then, there is a positive probability that only malicious workers are selected and the master would have to audit always to obtain the correct result.
All previous works assume the existence of a fixed/predefined set of workers that the master always contacts. In this work we consider the existence of a pool of $N$ workers out of which the master chooses $n<N$. 
\vspace*{-2em}
\paragraph{Our contributions.} 

\begin{itemize}\vspace*{-.8em}
\item 
We present a mechanism (in Section~\ref{rep:mechanism}) where the master chooses the most reputable workers for each round of computation, allowing the system to eventually converge to a state where the correct result will be always obtained, with minimal auditing.
Our mechanism does not require workers to be available all the time.
To cope with the unavailability of the workers, we introduce a \emph{responsiveness reputation} that conveys the percentage of task assignments to which the worker replies with an answer.
The responsiveness reputation is combined with a \emph{truthfulness reputation} that conveys the reliability of the worker.
We enrich our study considering three types of truthfulness reputation. Namely, \typeD reputation (inspired in the ``adaptive replication" of BOINC), \typeB reputation (that we presented in~\cite{opodis2013}), and \typeA reputation (inspired on the work of Sonnek et al.~\cite{sonnek07}).
\item 
We also show formally (in Section~\ref{sec:analysis}) negative and positive results regarding the feasibility of achieving correctness in the long run in the absence of rational workers. Specifically, we show configurations (worker types, availability, etc.) of the pool of workers such that correctness cannot be achieved unless the master always audits, and the existence of configurations such that eventually correctness is achieved forever with minimal auditing.
\item 
We evaluate experimentally (in Section~\ref{Ncn:simulation}) our mechanism with extensive simulations under various conditions. Our simulations complement the analysis taking into account scenarios where rational workers exist. The different reputation types are compared showing trade-offs between reliability and cost.
\vspace*{-1em}
\end{itemize}


%% file: model-short.tex

\section{Model}
\label{rep:model}
\vspace*{-.5em}

\paragraph{\bf \em Master-Worker Framework.}
\noindent
We consider a master and a pool (set) of workers ${\cal N}$, where $|{\cal N}|=N$. The computation is broken into
{\em rounds} $r=1, 2,...$. In each round $r$, the master selects a set $W^{r}$ of $n < N$ workers, and sends them a task. The workers in $W^{r}$ are supposed to compute the task and return the result, but may not do so (e.g., unavailable computing other task). 
The master, after waiting for a fixed time $t$, proceeds with the received replies. Based on those replies, 
the master must decide which answer to take as the correct result for this round. 
The master employs a reputation mechanism put in place to choose the $n$ most reputable workers in every round.
We assume that tasks 
have a unique solution; although such limitation reduces the scope of application of the 
presented mechanism~\cite{TACB05}, there are plenty of computations where the correct 
solution is unique: e.g., any mathematical function.\vspace{-.8em} 


\paragraph{\bf \em Worker unavailability.} In Internet-based master-worker computations, and especially in volunteering computing, workers
are not always available to participate in a computation~\cite{Heienetal09} (e.g., they are off-line for a particular period of time).  
We assume that each worker's availability is stochastic and independent of other workers. 
Formally, we let $d_i >0$ be the probability that the master receives the reply from worker $i$ within time $t$ 
(provided that the worker was chosen by the master to participate in the computation for the given round $r$, i.e., $i\in W^{r}$). In other words, this is
the probability that the worker is available to compute the task assigned.\vspace{-1em}  


\paragraph{\bf \em Worker types.} We consider three types of workers: {\em rational, altruistic,} and {\em malicious}.
Rational workers are selfish in a game-theoretic sense and their aim is to maximize their utility (benefit). In the context of this paper, a worker is {\em honest} 
in a round, when it truthfully computes and returns the correct result, and it {\em cheats} when it returns some incorrect value. 
Altruistic and malicious workers have a predefined behavior: to always be honest and cheat respectively.
Instead, a rational worker decides to be honest or cheat depending on which strategy maximizes its utility. We denote by $p_{Ci}(r)$ the probability of a rational worker $i$ cheating in round $r$,
provided that $i\in W^{r}$.  
The worker adjusts this probability over the course of the multiround computation
using a reinforcement learning approach. 
The master is not aware of each worker type, neither of the distribution over types. That is, our mechanism does not rely on any statistical information. 

While workers make their decision individually and with no coordination, following \cite{Sarmenta02,PPL12}, we assume that all the workers that cheat in a round return the same incorrect value. 
This yields a worst case scenario for the master to obtain the correct result using a voting mechanism. 
This assumption subsumes models where cheaters do not necessarily return the same answer, and it can
be seen as a weak form of collusion.\vspace{-.8em} 


\paragraph{\bf \em Auditing, Payoffs, Rewards and Aspiration.}
\noindent
When necessary, the master employs {\em auditing} and
\emph{reward/punish} schemes to induce the rational workers to be honest. 
In each round, the master may decide to audit the response of the workers, 
at a cost. 
In this work, auditing means that the master computes the task by itself, and checks
which workers have been honest. We denote by $p_\VRF(r)$ the probability of the master auditing the 
responses of the workers in round $r$. The master can change this auditing probability over the course of the computation,
but restricted to a minimum value $p_\VRF^{min}>0$.
When the master audits, it can accurately reward and punish workers. 
When the master does not audit, it rewards only those in the weighted majority (see below)
of the replies received and punishes no one. 

We consider three worker payoff parameters: (a)$\Cp$: worker's punishment for being caught cheating,
(b) $\Ct$: worker's cost for computing a task, and (c) $\SW$: worker's benefit (typically payment) from the
master's reward. 
As in~\cite{BM55}, we also assume that a worker $i$ has an {\em aspiration} $a_i$, which is the minimum benefit that worker $i$ expects to obtain in a round. 
We assume that the master has the freedom of choosing $\SW$ and $\Cp$ with the
goal of satisfying {\em eventual correctness}, defined next.
E.g., in order
to motivate the worker to participate in the computation, the master ensures that 
$\SW - \Ct \geq a_i$; in other words, the worker has the potential of its aspiration to
be covered even if it computes the task.\vspace{-.8em} 

\paragraph{\bf \em Eventual Correctness.} \noindent
The goal of the master is to eventually obtain a reliable computational platform: After some 
finite number of rounds, the system must guarantee that the master obtains the correct task results 
in every round with probability $1$ and audits with probability $p_\VRF^{min}$. 
We call such property \emph{eventual correctness}.
Observe that eventual correctness implies that eventually the master receives at least one (correct) reply in every round.  \vspace{-.8em}

\paragraph{\bf \em Reputation.} \noindent
The reputation of each worker is measured and maintained by the master.
Reputation is used by the master to cope with the uncertainty about the workers' truthfulness and availability. 
In fact, the workers are unaware that a reputation scheme is in place, and their interaction with 
the master does not reveal any information about reputation; i.e., the payoffs do not depend on a worker's reputation. 
The master wants to assign tasks to workers that are reliable, that is, workers that are both responsive {\em and} truthful.
Hence, we consider the worker's reputation as the product of two factors: responsiveness reputation and truthfulness reputation.
Thus, the malicious workers will obtain a low reputation fast due to their low truthfulness reputation, and also the workers that are generally unavailable will get a low reputation due to their low responsiveness reputation. Consequently, these workers will stop being chosen by the master.

More formally, we define the reputation of a worker $i$ as $\rho_i = \rho_{rs_i} \cdot \rho_{tr_i}$, where 
$\rho_{rs_i}$ represents the responsiveness reputation  and 
$\rho_{tr_i}$ the truthfulness reputation of worker $i$. 
We also define the reputation of a set of workers $Y\subseteq W$ as the aggregated reputation of all workers in $Y$. That is, $\rho_Y(r)=\sum_{i \in Y} \rho_i(r)$.

In this work, we consider three truthfulness reputation types: \typeA, \typeB, and \typeC. 
In the \typeA reputation type (introduced in~\cite{sonnek07}) the reputation changes at
a linear rate. 
The \typeB reputation type (introduced in~\cite{opodis2013}) is ``unforgiving'', in the sense that the reputation of a worker caught cheating will never increase. The reputation of a worker in this type changes at an exponential rate. 
The \typeC reputation type is inspired by BOINC~\cite{boinc_reputation}. In the BOINC system this reputation method is used to avoid redundancy if a worker is considered honest\footnote[2]{In BOINC, honesty means that the worker's task result agrees with the majority, while in our work this decision is well-founded, since the master audits.}.
For the responsiveness reputation we use the \typeA reputation, adjusted for responses. 
For the worker's availability it is natural to use a ``forgiving" reputation, especially when considering volunteer computing.    
For the detailed description of the reputation types we introduce some necessary notation as follows. \vspace*{-.5em}
\begin{itemize}[leftmargin=2mm]
\item [] {$select_i(r)$:} the number of rounds the 
master selected worker $i$ up to round $r$.
\item [] {$reply\_select_i(r)$:} the number of rounds up to round $r$ in which worker $i$ was selected and the master received a reply from $i$.
\item [] {$audit\_reply\_select_i(r)$:} the number of rounds up to round $r$ where the master selected worker $i$, received its reply
and audited.
\item [] {$correct\_audit_i(r)$:} the number of rounds up to round $r$ where
the master selected worker $i$, received its reply, audited and $i$ was truthful. 
\item [] {$streak_i(r)$:} the number of rounds $\leq r$ in which worker $i$ was selected, audited, and replied correctly after the latest round in which it was selected, audited, and caught cheating.
\end{itemize}\vspace*{-.5em}

%
%
Then, the reputation types we consider are as follows.
\begin{framed}
\begin{small}
\noindent
{\bf Responsiveness reputation:}
$\rho_{rs_i}(r) = \frac{reply\_select_i(r)+1}{select_i(r)+1}.$
\newline
{\bf Truthfulness reputation:}
\begin{itemize}[leftmargin=2mm]
\item 
[] {\bf \typeA:} $\displaystyle
\rho_{tr_i}(r) = \frac{correct\_audit_i(r)+1}{audit\_reply\_select_i(r)+1}.
$
\item
[] {\bf \typeB:} 
 $\displaystyle \rho_{tr_i}(r) =  \varepsilon^{audit\_reply\_select_i(r)-correct\_audit_i(r)},$
where $\varepsilon \in (0,1)$.
%
%
%

\item [] {\bf \typeC:} 
$
  \rho_{tr}(r)=\begin{cases}
    0, & \text{if } streak(r)<10.\\
    1-\frac{1}{streak(r)}, & \text{otherwise}.
  \end{cases}$
\end{itemize}\vspace*{-1.5em}
\end{small}
\end{framed}

All workers are assumed to have the same initial reputation before the master interacts with them. 
The goal of the above definitions is for workers who are responsive {\em and} truthful to eventually have high reputation, whereas workers who are not responsive {\em or} not truthful, to eventually have low reputation. 


%% file: reputation.tex
\vspace*{-3mm}
\section{Reputation-based Mechanism}
\label{rep:mechanism}
\vspace*{-5mm}
We now present our reputation-based mechanism. 
The mechanism is composed by an algorithm run by the master and an algorithm
run by each worker.\vspace{-.7em} 

\paragraph{\bf \em Master's Algorithm.} \noindent
The algorithm followed by the master, Algorithm~\ref{alg4}, begins by choosing the initial probability of auditing and the initial reputation (same for all workers). 
The initial probability of auditing will be set according to the information the master has
about the environment (e.g., workers' initial $p_C$). 
For example, if it
has no information about the environment, a natural approach would be to initially set $p_\VRF=0.5$ or $p_\VRF=1$ (as a more conservative approach).
The master also chooses the truthfulness reputation type to use.\vspace{1em} 

\lstset{columns=fullflexible,tabsize=3,basicstyle=\footnotesize,identifierstyle=\rmfamily\textit,
mathescape=true,morekeywords={choose,process,const,when,elsif,procedure,null,function,do,return,
foreach,begin,var,if,then,else,rcv,rx_user_events,rx_network_events,for,end,receive,send,upstream,downstream,
while,do,forward,backward,upon,wait,audit,update,accept},literate={:=}{{$\leftarrow$ }}1{->}{{ $\rightarrow$ }}1}
\lstset{escapeinside={('}{')}}
\newlength\listingnumberwidth
\setlength\listingnumberwidth{\widthof{00} + 0.5em}
\lstset{numbers=left, xleftmargin=\listingnumberwidth, numbersep=1em}
\begin{figure}[t!]
\vspace{-1em}
\begin{algorithm}[H]
\caption{Master's Algorithm}
\label{alg4}\vspace{-.8em}
\begin{lstlisting}
$p_\VRF$ := $x$, where $x \in [p_\VRF^{min}, 1]$
for $i$ := $0$ to $N$ do
    $select_i$ := $0$;$~reply\_select_i$:=$0$;$~audit\_reply\_select_i$:=$0$;$~correct\_audit_i$ := $0; streak_i\gets 0$
    $\rho_{{rs}_i}$ := $1$; initialize $\rho_{{tr}_i}$		// initially all workers have the same reputation
for $r$ := $1$ to  $\infty$ do   
	$W^r$ := $\{i\in {\cal N}: i$ is chosen as one of the $n$ workers with the highest $\rho_i=\rho_{{rs}_i}\cdot \rho_{{tr}_i}$  $\}$				
	$\forall i\in W^r: select_i$ := $select_i + 1$
	send a task $T$ to all workers in $W^r$
	collect replies from workers in $W^r$ for $t$ time
	wait for $t$ time collecting replies as received from workers in $W^r$
	$R$ := $\{i \in W^r :$ a reply from $i$ was received by time $t \}$  		
	$\forall i\in R : reply\_select_i$ := $reply\_select_i + 1$
   update responsiveness reputation $\rho_{rs_i}$ of each worker	$i\in W^r$ 	
   audit the received answers with probability $p_\VRF$
   if the answers were not audited then
      accept the value $m$ returned by workers $R_m \subseteq R$,
            where $\forall m', \rho_{tr_{R_m}} \geq \rho_{tr_{R_{m'}}}$			// weighted majority of workers in $R$
   else 		 		// the master audits
      foreach $i\in R$ do
      	$audit\_reply\_select_i$ := $audit\_reply\_select_i + 1$    	
      	if $i \in F$ then $streak_i$:=$0$ 			// $F\subseteq R$ is the set of responsive workers caught cheating
      	else   $correct\_audit_i$ :=$correct\_audit_i + 1$, $streak_i$:=$streak_i+1$			// honest responsive workers
      	update truthfulness reputation $\rho_{tr_i}$		 // depending on the type used
      if $\rho_{{tr}_R}=0$ then $p_\VRF$:=$\min\{1,p_\VRF+\alpha_m\}$
      else
          $p_\VRF'$ := $p_\VRF+ 
                               \am(\rho_{{tr}_F} / \rho_{{tr}_R}-\tau)$ 
          $p_\VRF$ := $\min\{1, \max\{p_\VRF^{min}, p_\VRF'\}\}$ 
	$\forall i\in W^r :$ return $~\Pi_i~\mathit{to}$	worker $i$ 					// the payoff  of workers in $W^r \setminus R$ is zero 
\end{lstlisting}\vspace{-.5em}
\end{algorithm}
\hfill\vspace{-4.5em}
\end{figure}\vspace{-1em}

\begin{figure}[h!]
\vspace{-1.5em}
\begin{algorithm}[H]
\caption{Algorithm for Rational Worker $i$}
\label{alg2}\vspace{-.5em}
\begin{lstlisting}
$p_{Ci}$ := $y$, where $y \in [0, 1]$
repeat forever
    wait for a task $T$ from the master
    if available then
        decide whether to cheat or not independently with distribution $P(cheat)=p_{Ci}$  
        if the decision was to cheat then
           send arbitrary solution to the master
           get payoff $\Pi_i$
           $p_{Ci}$ := $\max\{0,\min\{1,p_{Ci} + \aw(\Pi_i-a_i)\}\}$
        else 
           send $compute(T)$ to the master
           get payoff $\Pi_i$
           $p_{Ci}$ := $\max\{0,\min\{1,p_{Ci} - \aw(\Pi_i - \Ct-a_i)\}\}$   
\end{lstlisting}\vspace{-.5em}
\end{algorithm}\vspace{-4em}
\end{figure}

At the beginning of each round, the master chooses the $n$ most reputable workers out of the total $N$ 
workers (breaking ties uniformly at random) and sends them a task $T$. 
In the first round, since workers have the same reputation, the choice is uniformly at random.
Then, after waiting $t$ time to receive the replies from the selected workers, the master proceeds with the mechanism.
The master updates the responsiveness reputation and audits the answers with probability $p_\VRF$. 
In the case the answers are not audited, the master accepts the value returned by the weighed majority.
In Algorithm~\ref{alg4}, $m$ is the value returned by the weighted majority and $R_m$ is the subset of workers that returned $m$. 
If the master audits, it updates the truthfulness reputation and the audit probability for the next round.  
Then, the master rewards/penalizes the workers as follows. If the master audits and a worker $i$ is a cheater (i.e., $i\in F$), then $\Pi_i = -\Cp$; if $i$ is honest, then $\Pi_i = \SW$. If the master does not audit, and $i$ returns the value of the weighted majority
 (i.e., $i\in R_m$), then
$\Pi_i = \SW$, otherwise $\Pi_i=0$.

In the update of the audit probability $p_\VRF$, we include a threshold, denoted by $\tau$, that represents the master's \emph{tolerance} to cheating
(typically, we will assume $\tau=1/2$ in our simulations). If the ratio of the aggregated reputation of cheaters with respect to the total is larger than $\tau$, $p_\VRF$ is increased, and decreased otherwise.
The amount by which $p_\VRF$ changes depends on the difference between these
values, modulated by a {\em learning rate} $\am$~\cite{RLbook}. This latter value determines to what extent the newly acquired information will override the old information. For example, if $\am=0$ the master will never adjust $p_\VRF$.\vspace{-.7em} 



\paragraph{\bf \em Workers' Algorithm.} \noindent
Altruistic and malicious workers have predefined behaviors.
When they are selected and receive a task $T$ from the master, if they are available, they compute the task (altruistic) or fabricate an arbitrary solution (malicious), replying accordingly.  
 If they are not available, they do not reply.
Rational workers 
run the algorithm described in Algorithm~\ref{alg2}.
The execution of the algorithm begins with a rational worker $i$ deciding an initial probability of cheating $p_{Ci}$.  
Then, the worker waits to be selected and receive a task $T$ from the master. When so, and if it is available at the time, then with probability $1-p_{Ci}$, worker $i$ 
computes the task and replies to the master with the correct answer. Otherwise, 
it fabricates an answer, and sends the incorrect response to the master. 
After receiving its payoff, worker $i$ changes its $p_{Ci}$ according to payoff $\Pi_i$, the chosen strategy (cheat or not cheat), and its aspiration $a_i$. Similarly to the master, the workers have a {\em learning rate} $\aw$. We assume that all workers have the same learning 
rate, that is, they learn in the same manner (in~\cite{RLbook}, the learning rate is called step-size). \added[ev]{In a real platform the workers’ learning rate can slightly vary (since workers in these platforms have similar profiles), making some worker more
or less susceptible to reward and punishment. Using the same learning rate for all workers is representative of what happens in a population of 
different values with small variations around some mean.}\vspace{-1em}

%


%% file: analysis.tex

\section{Analysis}
\label{sec:analysis}
\vspace*{-5mm}
In this section, we prove some properties of the system. 
We start by observing that, in order to achieve eventual correctness, it is necessary to change workers over time.
{\em Omitted proofs are given in the Appendix.} 
\vspace*{-2mm}
\begin{observation}
\label{obs1}
If the number of malicious workers is at least $n$ and the master assigns the task to the same workers in all rounds, eventual correctness cannot be \replaced[ev]{guaranteed}{satisfied}.
\end{observation} \vspace*{-2mm}
\added[ev]{The intuition behind this observation is that there is always a positive probability that the master will select $n$ malicious workers at the first round and will have to remain with the same workers.}
This observation justifies that the master has to change its choice of workers if eventual correctness has to be guaranteed.  
We apply the natural approach of choosing the $n$ workers with the largest reputation among the $N$ workers in the pool (breaking ties randomly). 
In order to guarantee eventual correctness we need to add one more condition regarding the availability of the workers. 
\vspace*{-.8em}
\begin{observation}
\label{obs2}
\replaced[ev]{To guarantee eventual correctness at least one non-malicious worker $i$ must exist with $d_i=1$.}{To satisfy eventual correctness at least one worker $i$ that is not malicious must have $d_i=1$.}
\end{observation} \vspace*{-2mm}
\remove{
\begin{proof}
For the sake of contradiction assume that every non-malicious worker $i$ has $d_1<1$ and eventual correctness is satisfied. Then, by definition
of eventual correctness, there is a round $r_0$ such that in all rounds $r \geq r_0$ the master uses $p_\VRF=p_\VRF^{min}<1$.
In any round $r \geq r_0$ there is a positive probability that the master does not audit and all the replies received (if any) are incorrect. Then, there is a positive probability that the master does not obtain the correct task result, which is a contradiction.
\qed
\end{proof}
}

To complement the above observations, we show now that there are sets of workers with which eventual correctness is achievable using the different reputation types (\typeA and \typeB as truthfulness reputations) defined and the master reputation-based mechanism in Algorithm~\ref{alg4}.
\vspace*{-.6em}

\begin{theorem}
Consider a system in which \replaced[ev]{workers are either altruistic or malicious}{there are no rational workers} and there is at least one altruistic worker $i$ with $d_i=1$ in the pool. Eventual correctness is satisfied if the mechanism of Algorithm~\ref{alg4} is used with the responsiveness reputation and any of the truthfulness reputations \typeA or \typeB. 
\label{theorem_t1_t2}
\end{theorem} \vspace*{-2mm}

\added[ev]{The intuition behind the proof is that thanks to the decremental way in which the reputation of a malicious worker is calculated at some point the altruistic worker $i$ with full responsiveness ($d_i=1$) will be selected and have a greater reputation than the aggregated reputation of the selected malicious workers.}
A similar result does not hold if truthfulness reputation of type \typeD is used. In this case, we have found that it is not enough that one altruistic worker with 
full availability 
exists, but also the number of altruistic workers with partial availability 
have to be considered.
\vspace*{-.8em}
\begin{theorem}
Consider a system in which \replaced[ev]{workers are either altruistic or malicious}{ there are  no rational workers} and there is at least one altruistic worker $i$ with $d_i=1$ in the pool. 
In this system, the mechanism of Algorithm~\ref{alg4} is used with the responsiveness reputation and the truthfulness reputation \typeD. 
Then, eventual correctness is satisfied if and only if the number of altruistic workers with $d_j <1$ is smaller than $n$.
\label{theorem_t4}
\end{theorem}
\begin{proof}
\vspace*{-2mm}
In this system, it holds that every malicious worker $k$ has truthfulness reputation $\rho_{tr_k}=0$ forever, since the replies that the master receives from it (if any) are always incorrect. Initially, altruistic workers also have zero truthfulness reputation. An altruistic worker $j$ has positive truthfulness reputation after it is selected, and its reply is received and audited by the master 10 times. Observe that, once that happens, the truthfulness reputation of worker $j$ never becomes 0 again. Also note that the reponsiveness reputation never becomes 0. Hence, the first altruistic workers that succeed in raising their truthfulness reputation above zero are always chosen in future rounds. While there are less than $n$ workers with positive reputation, the master selects at random from the zero-reputation workers in every round. Then, eventually (in round $r_0$) there are $n$ altruistic workers with positive reputation, or there are less than $n$ but all altruistic workers are in that set. After then, no new altruistic worker increase its reputation (in fact, is ever selected), and the set of altruistic selected workers is always the same.

If the number of altruistic workers with $d_j <1$ is smaller than $n$, since worker $i$ has $d_i=1$, after round $r_0$ among the selected workers there are altruistic workers with $d_j=1$ and positive reputation. Then, in every round there is going to be a weighted majority of correct replies, and eventual correctness is guaranteed.

If, on the other hand, the number of altruistic workers with $d_j <1$ is at least $n$, there is a positive probability that all the $n$ workers with positive reputation are from this set. \added[ev]{Since there is a positive probability that $n$ altruistic workers with  $d_j <1$ are selected in round $r_0$ with probability one the worker $i$ with $d_i=1$ will never be selected.} 
 If this is the case, eventual correctness is not satisfied \added[ev]{(since there is a positive probability that the master will not receive a reply in a round)}. Assume otherwise and consider that after round $r'_0$ it holds that $p_\VRF = p_\VRF^{min}$. Then, in every round after $r'_0$ there is a positive probability that the master receives no reply from the selected workers and it does not audit, which implies that it does not obtain the correct result.
\qed
\end{proof}
\vspace*{-.5em}
This result is rather paradoxical, since it implies that a system in which all workers are altruistic (one with $d_i=1$ and the rest with $d_j<1$) does not guarantee eventual correctness, while a similar system in which the partially available workers are instead malicious does.
\added[ev]{This paradox comes to stress the importance of selecting the right truthfulness reputation. Theorem~\ref{theorem_t4} shows a positive correlation among a truthfulness reputation with the availability factor of a worker in the case a large number of altruistic workers.}
\vspace{-1.5em}

%% file: simulations-new-arxiv-new.tex

\newcommand{\myboldmath}{}
\newcommand{\defn}[1]           {{\textit{\textbf{\myboldmath #1}}}}
\newcommand{\defi}[1]           {{\textit{\textbf{\myboldmath #1\/}}}}

\section{Simulations}
\label{Ncn:simulation}
\vspace*{-5mm}
Theoretical analysis is complemented with illustrative simulations on a number of different scenarios for the case of full and partial availability. The simulated cases give indications on the values of some parameters (controlled by the master, namely the type of reputation and the initial $p_\VRF$) under which the mechanism performs better. The rest of the parameters of the mechanism and the scenarios presented are essentially based on the observations extracted from~\cite{einstein,emBoinc}, and are rather similar to our earlier work~\cite{opodis2013}. 
We have developed our own simulation setup by implementing our mechanism (Algorithms \ref{alg4} and \ref{alg2}, and the reputation types discussed above) using C++. The simulations were executed on a dual-core AMD Opteron 2.5GHz processor, with 2GB RAM, running CentOS version 5.3.

For simplicity,  we consider that all workers have the same aspiration level $a_i= 0.1$, although we have checked that with random values the results are similar to those presented here, provided their variance is not very large \added[ev]{($ a_i \pm 0.1$)}.
We consider the
same learning rate for the master and the workers, i.e., $\alpha = \am = \aw=0.1$.
Note that the
learning rate, as discussed for example in~\cite{RLbook} (called step-size there), is generally set to a small constant value for practical reasons. 
We set $\tau=0.5$ 
(c.f., Sect.~\ref{rep:mechanism}; also see~\cite{CCPE13}), 
$p_\VRF^{min}=0.01$, and $\varepsilon=0.5$ in reputation \typeB. We assume that the master does not punish the workers $\Cp=0$, since depending on the platform used this might not be feasible, and hence more generic results are considered. Also we consider that the cost of computing a task is $\Ct=0.1$ for all workers and, analogously, the master is rewarding the workers with $\SW=1$ when it accepts their result \added[ev]{(for simplicity no further correlation among these two values is assumed)}. The initial cheating probability used by rational workers is $p_{Ci}=0.5$ and the number of selected workers is set to $n=5$.    

The first batch of simulations consider the case when the workers are fully available (i.e, all workers have $d=1$), and the behavior of the mechanism under different pool sizes is studied. The second batch considers the case where the workers are partially available. \vspace{-.8em}

\paragraph{\bf \em Full Availability.}
Assuming full worker availability we attempt to identify the impact of the pool size on different metrics: (1) the number of rounds, (2) number of auditing rounds, and (3) number of incorrect results accepted by the master, all of them measured until the system reaches convergence (the first round in which $p_\VRF=p_\VRF^{min}$)\added[ev]{\footnote{As we have seen experimentally, first the system reaches a reliable state and then $p_\VRF=p_\VRF^{min}$.}}.
Additionally, we are able to compare the behavior of the three truthfulness reputation types, showing different trade-off among reliability and cost.

\begin{figure}[t!]
\begin{center}
$
\begin{array}{ccc}
\hspace{-10em}
\includegraphics[width=2.3in, trim = 1mm 0mm 1mm 2mm, clip]{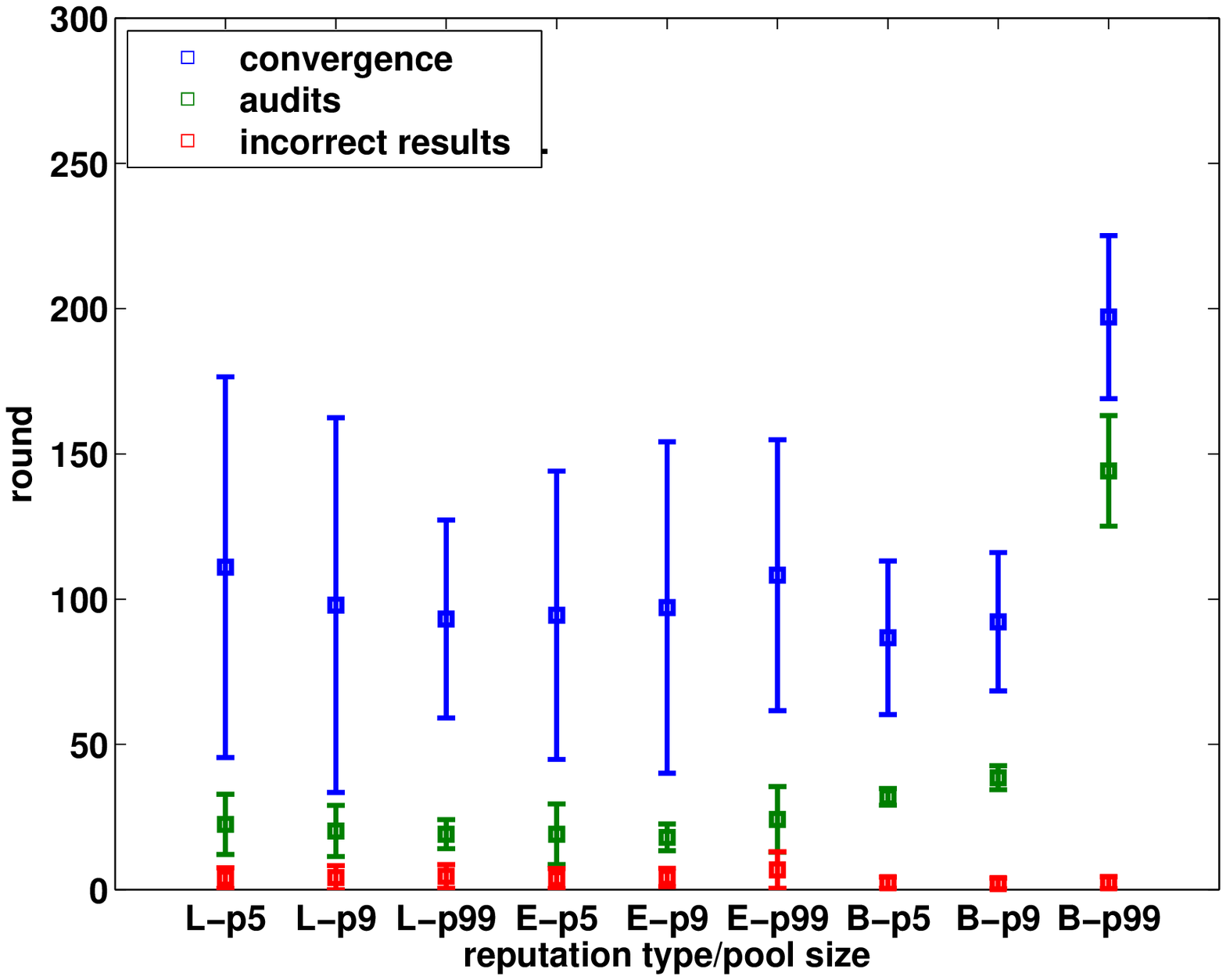}&
\includegraphics[width=2.3in, trim = 1mm 0mm 1mm 2mm, clip]{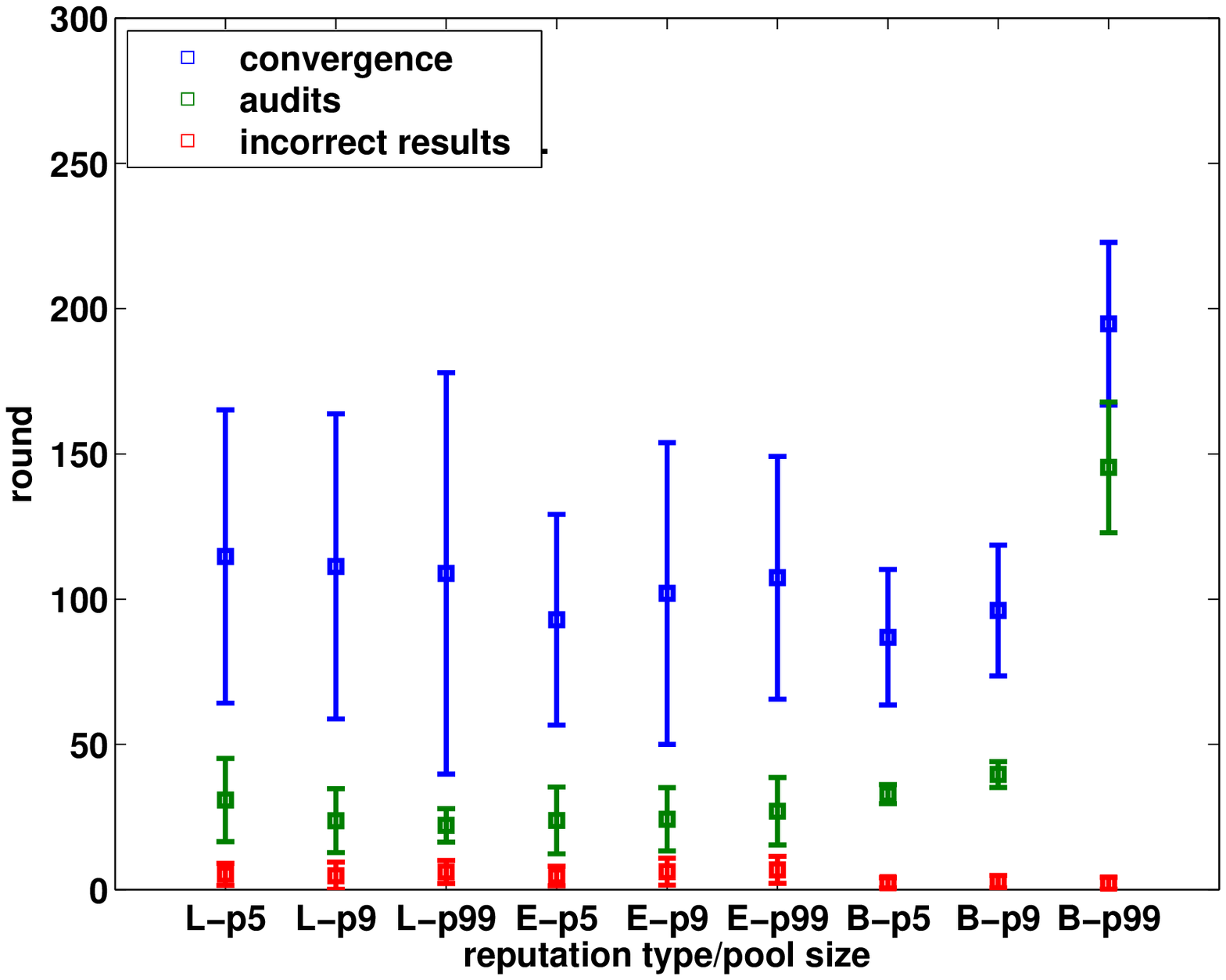}&
\includegraphics[width=2.3in, trim = 1mm 0mm 1mm 2mm, clip]{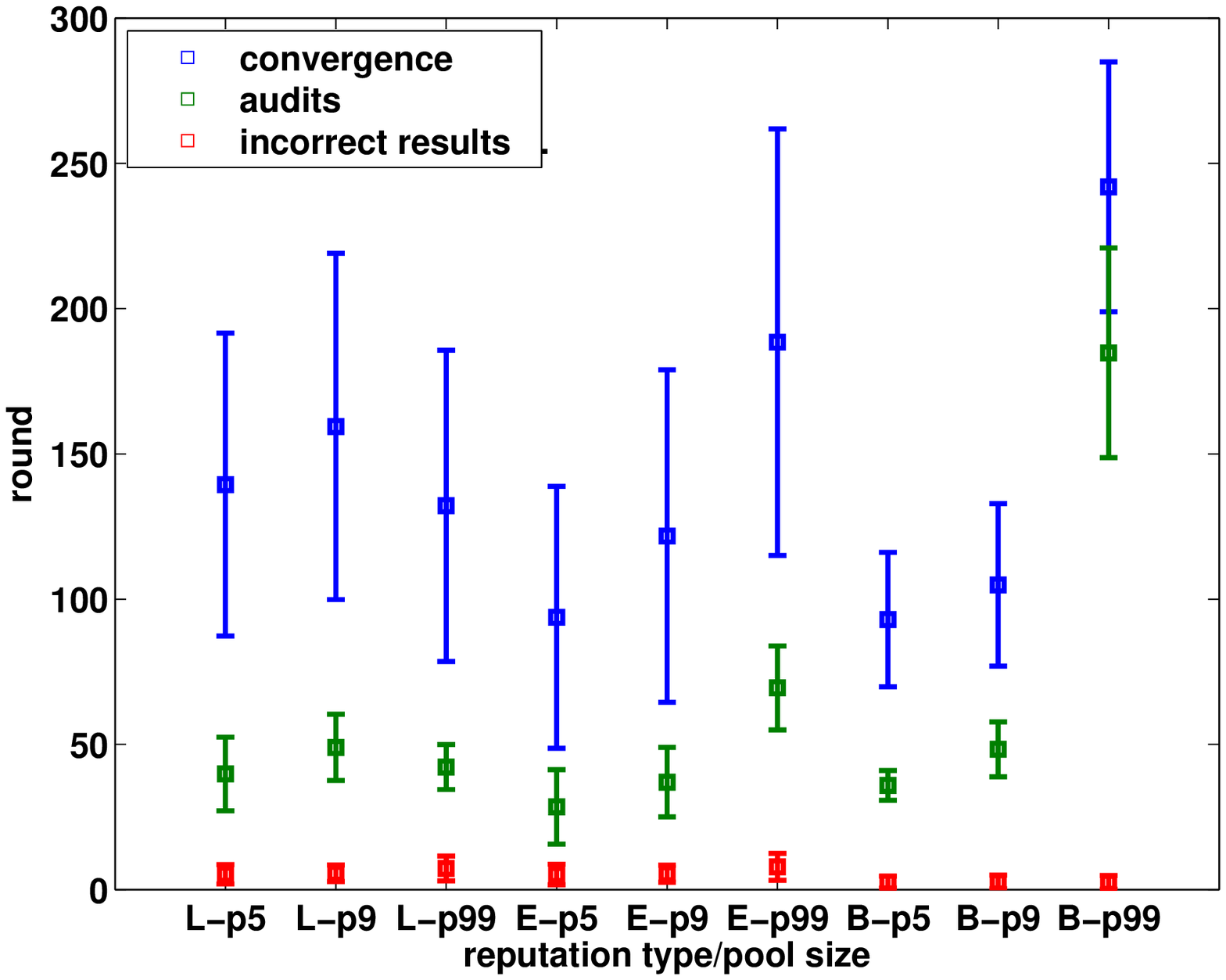}\\
\hspace{-8em}
(a1)&(b1)&(c1)\\
\hspace{-10em}
\includegraphics[width=2.3in, trim = 1mm 0mm 1mm 2mm, clip]{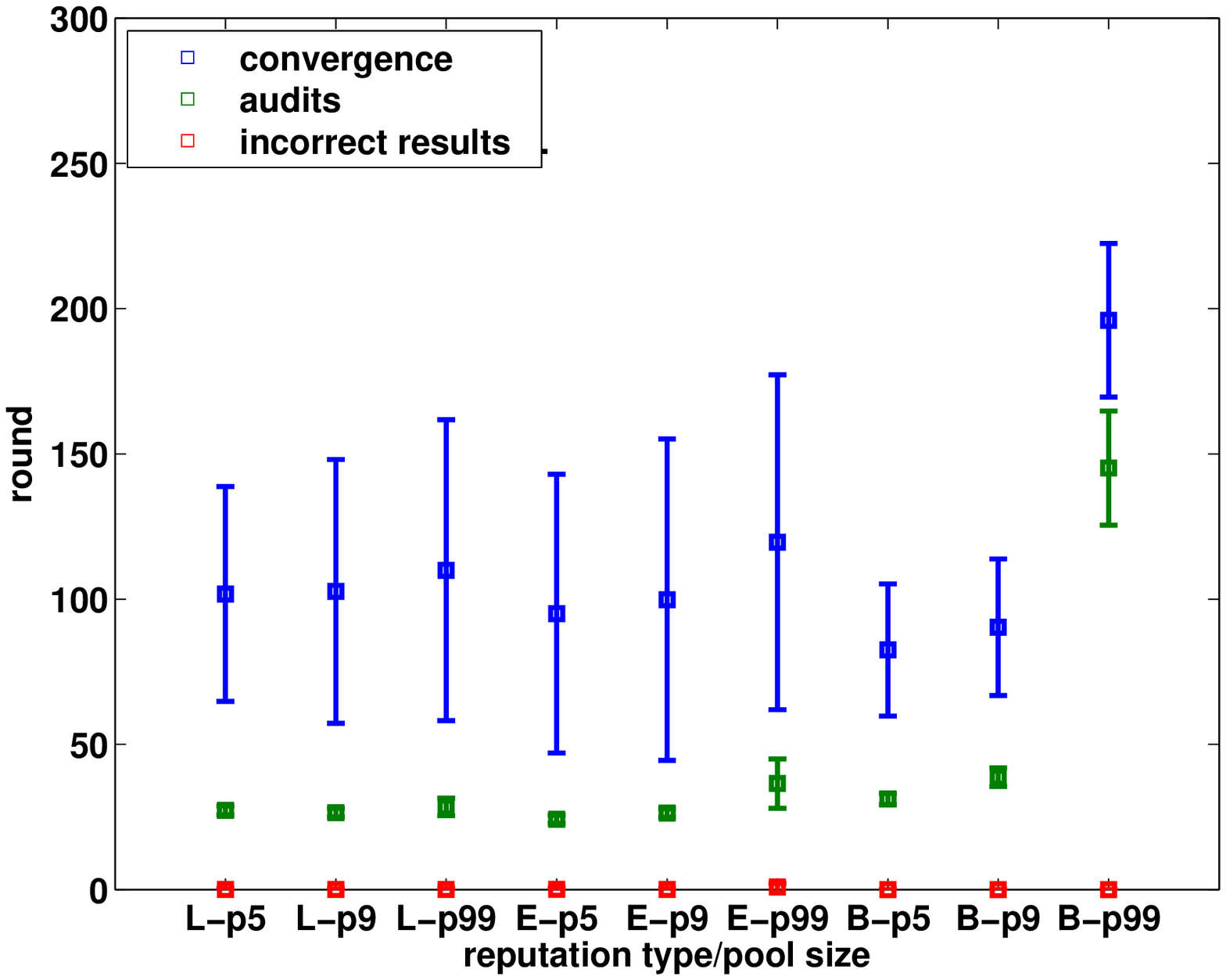}&
\includegraphics[width=2.3in, trim = 1mm 0mm 1mm 2mm, clip]{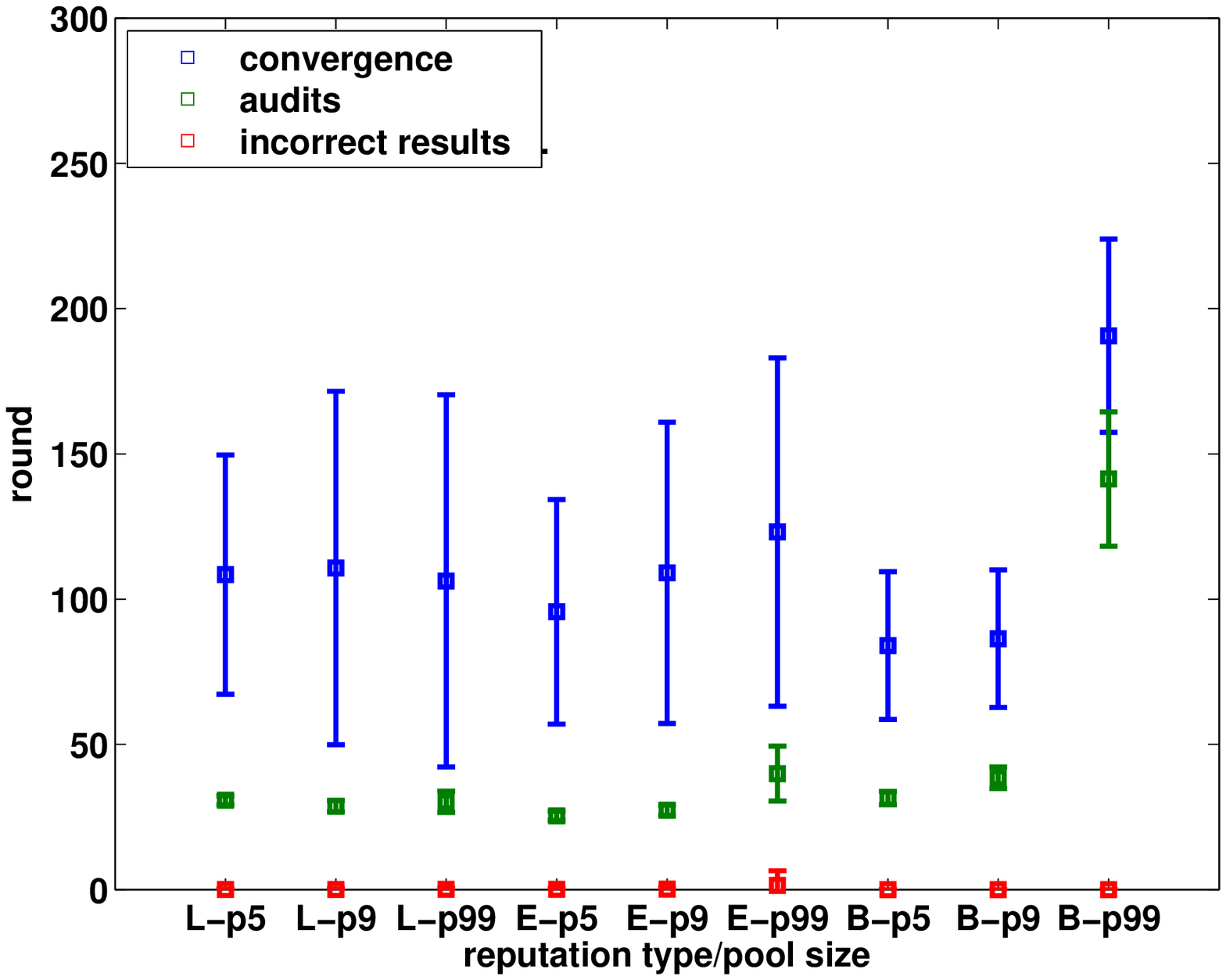}&
\includegraphics[width=2.3in, trim = 1mm 0mm 1mm 2mm, clip]{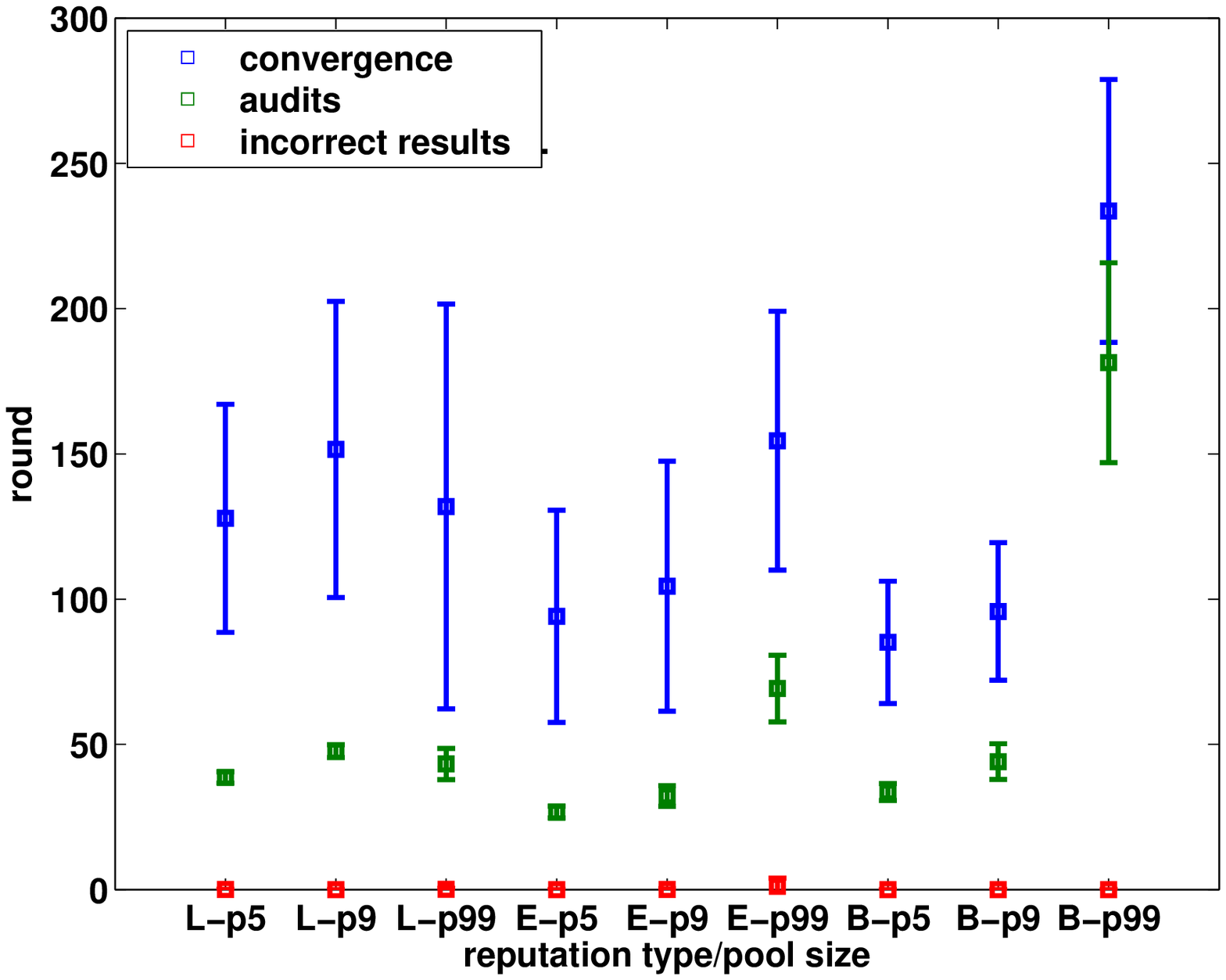}\\
\hspace{-8em}
(a2)&(b2)&(c2)\vspace{-1em}\\
\end{array}$
\caption{\small Simulation results with full availability. First row plots are for initial $p_\VRF=0.5$. Second row plots are for initial $p_\VRF=1$. The bottom (red) errorbars present the number of incorrect results accepted until convergence ($p_\VRF=p_\VRF^{min}$), the middle (green) errorbars present the number of audits until convergence; and finally the upper (blue) errorbars present the number of rounds until convergence, in 100 instantiations. In plots (a1) and (a2) the ratio of rational/malicious is 5/4. In plots (b1) and (b2) the ratio of rational/malicious is 4/5. In plots (c1) and (c2) the ratio of rational/malicious is 1/8. The x-axes symbols are as follows, L:~\typeA , E:~\typeB and B:~\typeD reputation; p5: pool size 5, p9: pool size 9 and p99: pool size 99.}
\label{rel-pa05}
\end{center}\vspace{-3.7em}
\end{figure}

We have tested the mechanism proposed in this paper with different initial $p_\VRF$ values. We present here two interesting cases of initial audit probability, $p_\VRF=0.5$
and $p_\VRF=1$. The first row of Figure~\ref{rel-pa05} (plots (a1) to (c1)) presents the results obtained in the 
simulations with initial $p_\VRF=0.5$ and the second row (plots (a2) to (c2)) the case
$p_\VRF=1$. The simulations in this section have been done for systems with only rational and malicious workers, with 3 different ratios between these worker types (ratios 5/4, 4/5, and 1/8), with different pool sizes ($N=\{5, 9, 99\}$), and for the 3 truthfulness reputation types. \added[ev]{These ratios consider the three most ``critical'' cases in which malicious workers can influence the results.}

A general conclusion we can extract from the first row of Figure~\ref{rel-pa05} (plots (a1) to (c1)) is that, independently of the ratio between malicious and rational workers, the trend that each reputation type follows for each of the different pool size scenarios is the same. (When the ratio of rational/malicious is 1/8 this trend is more noticeable.) 
Reputation \typeA does not show a correlation between the pool size and the evaluation metrics. This is somewhat surprising given that other two reputation types are
impacted by the pool size.
 
For reputation \typeB and \typeD we can observe that, as the pool size increases, the number of rounds until convergence also increases.
It seems like, for these reputation types, many workers from the pool have to be selected and audited before convergence. Hence, with a larger pool it takes more rounds for the mechanism to select and audit these workers, and hence to establish valid reputation for the workers and to reinforce the rational ones to be honest. 
For both reputation types (\typeB and \typeD) this is a costly procedure also in terms of auditing for all rational/malicious ratios. (The effect on the number of audits is more acute for reputation \typeD as the pool size increases.) As for the number of incorrect results accepted until convergence, with reputation \typeB they still increase with the
pool size. However, reputation \typeD is much more robust with respect to this metric, essentially guaranteeing that no incorrect result is accepted.

Comparing now the performance of the different reputation types based on our evaluation metrics, it seems that reputation \typeA performs better when the size of the pool is big compared to the other two reputation types. On the other hand reputation types \typeB and \typeD perform slightly better when the pool size is small. Comparing reputation types 
\typeB and \typeD, while reputation \typeD shows that has slightly faster convergence, this is traded for at least double auditing than reputation \typeB.
On the other hand, reputation \typeB is accepting a greater number of incorrect results until convergence. This is a clear example of the trade-off between convergence time, number of audits, and number of incorrect results accepted.      

Similar conclusions can be drawn when the master decides to audit with $p_\VRF=1$ initially, see Figure~\ref{rel-pa05}~(a2)~-~(c2). The only difference is that the variance, of the different instantiations on the three metrics is smaller. Hence, choosing $p_\VRF=1$ initially is a ``safer'' strategy for the master.\vspace{-.8em} 

\paragraph{\bf \em Partial Availability.}
Assuming now partial worker availability (i.e, workers may have $d<1$), we attempt to identify the impact of the unavailability of a worker on four different metrics: (1) the number of rounds, (2) number of auditing rounds, and (3) number of incorrect results accepted by the master, all until the system reaches convergence. In addition, we obtain (4) the number of incorrect results accepted by the master \emph{after} the system reaches convergence (which was zero in the previous section). Moreover, we are able to identify how suitable each reputation is, under different workers' ratio and unavailability probabilities. 

We keep the pool size fixed to $N=9$, and the number of selected workers fixed to $n=5$; and we analyze the behavior of the system in a number of different scenarios where the workers types and availabilities vary. The depicted scenarios present the cases of initial audit probability: $p_\VRF=\{0.5, 1\}$. 

\begin{figure}[t!]
\centering
$
\begin{array}{cc}
\hspace{-2em}
\includegraphics[width=2.5in, trim = 1.1mm 0mm 1mm 2mm, clip]{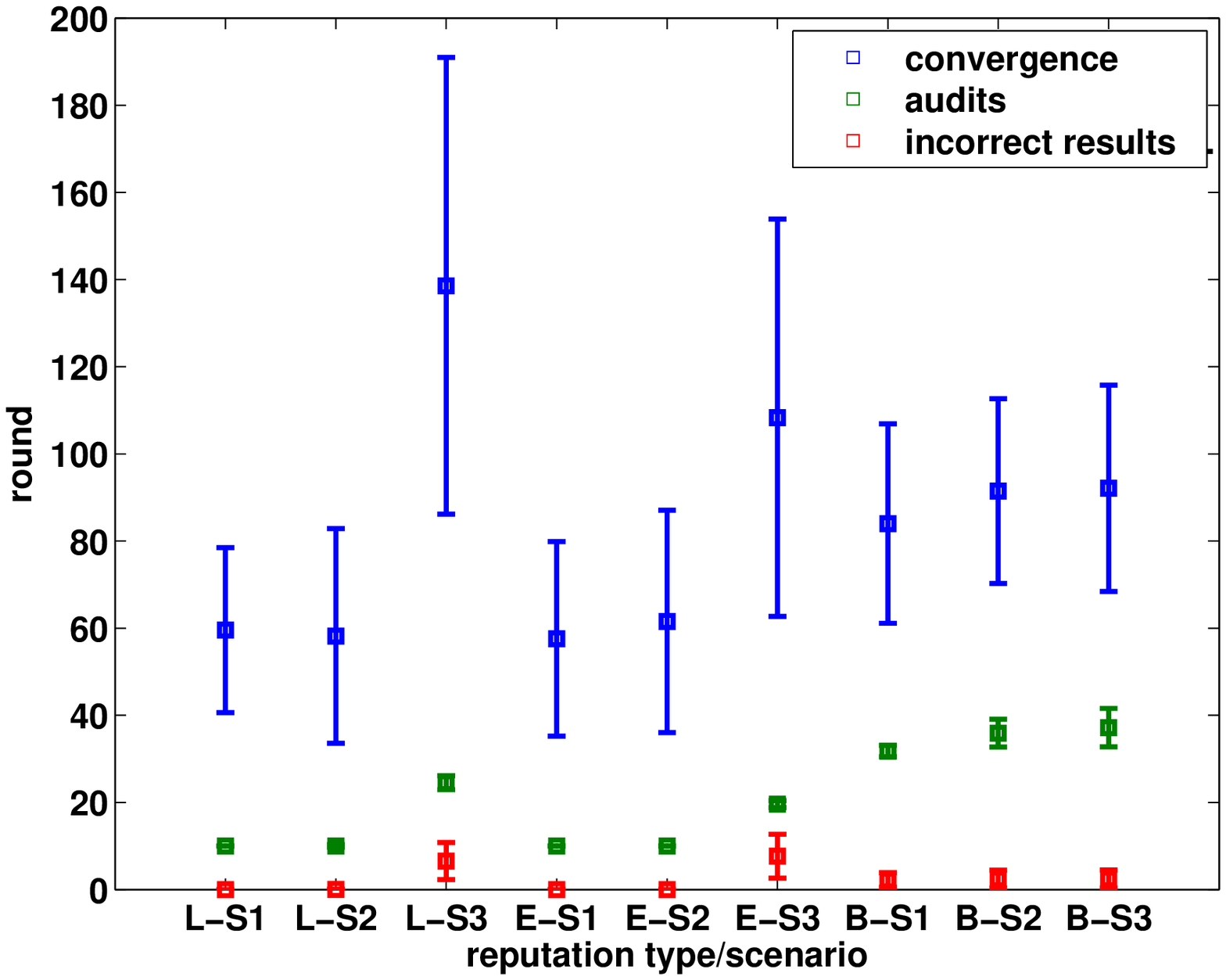}&
\includegraphics[width=2.5in, trim = 1.1mm 0mm 1mm 2mm, clip]{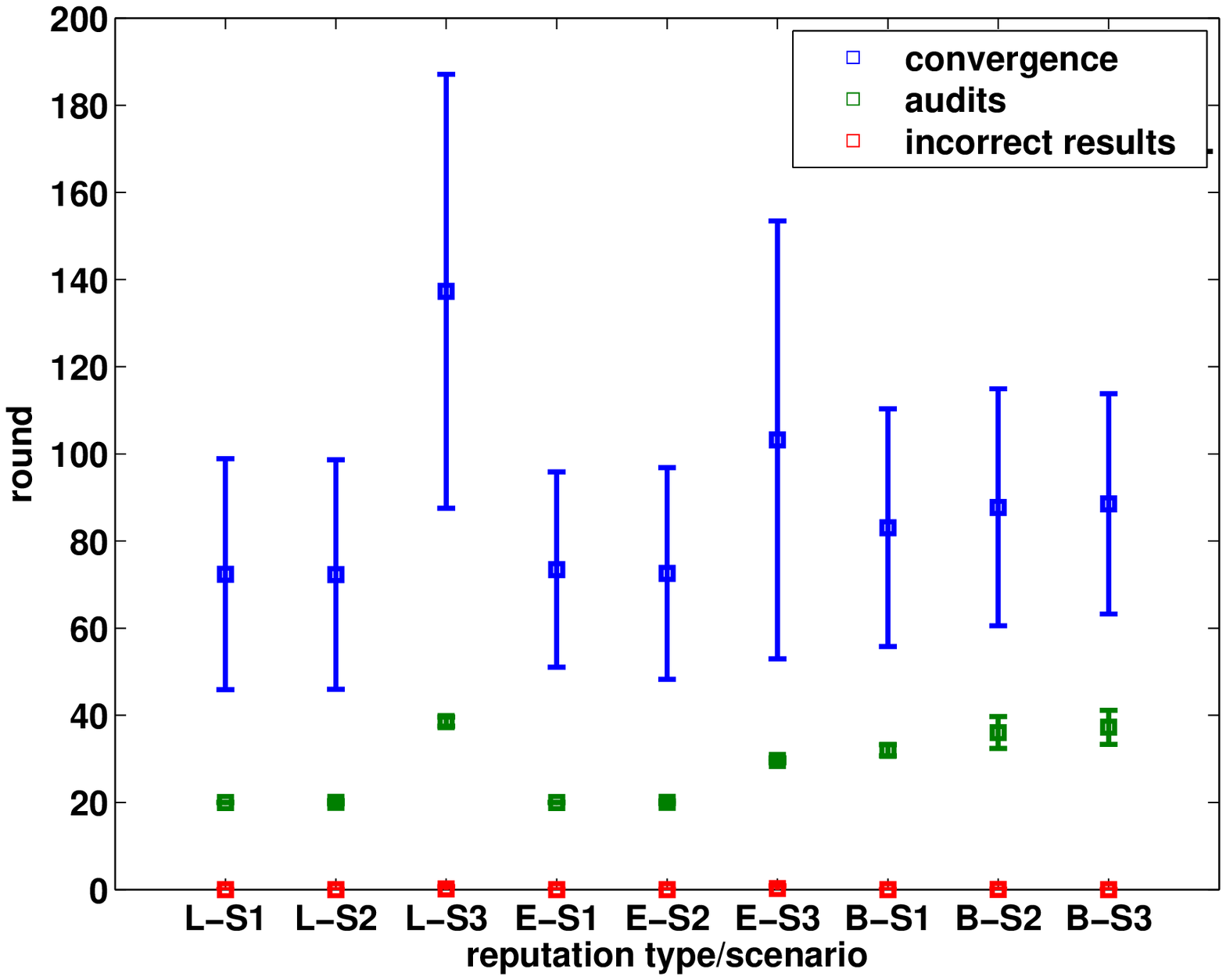}\\
(a1)&(b1)\\
\hspace{-2em}
\includegraphics[width=2.5in, trim = 1.1mm 0mm 1mm 2mm, clip]{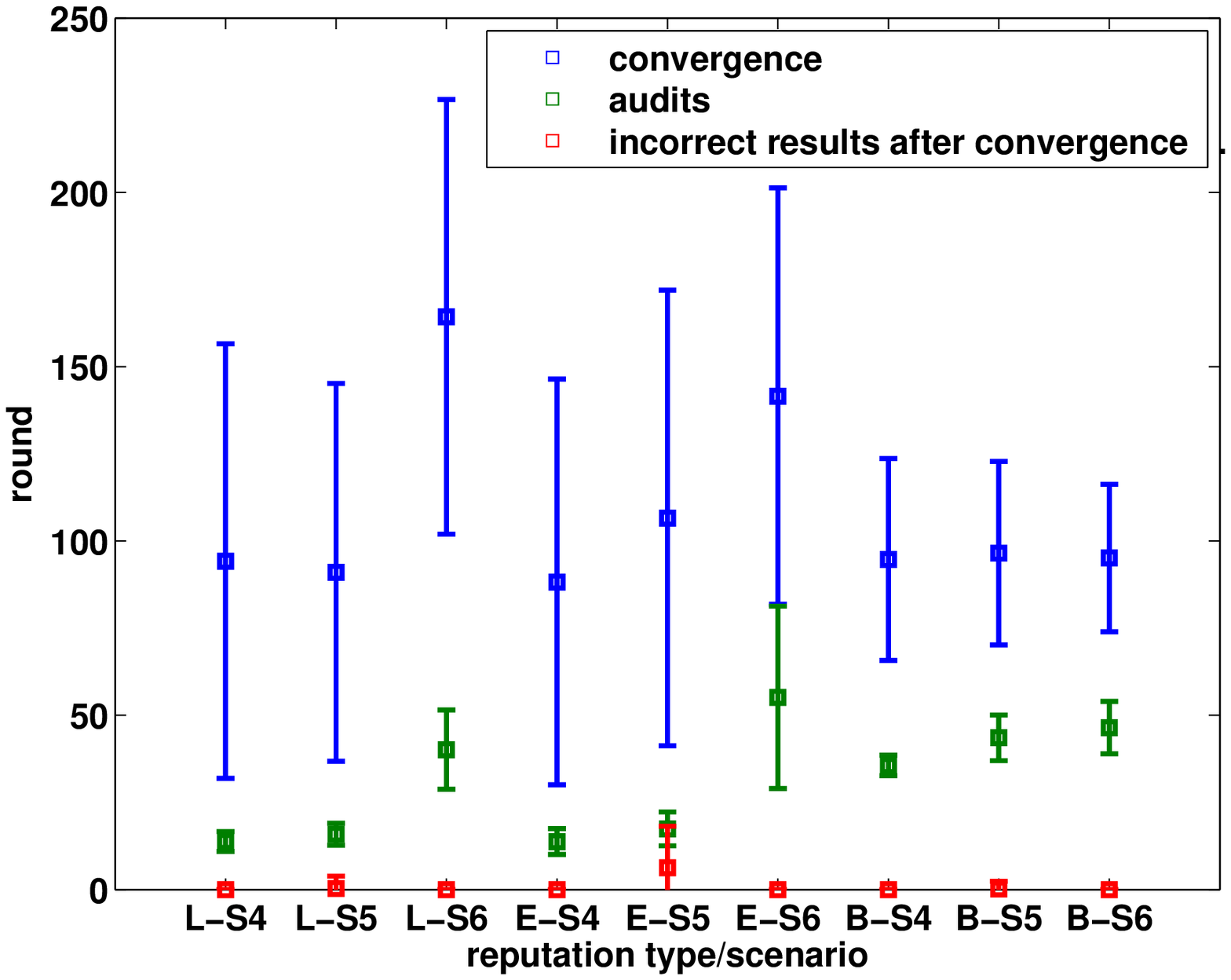}&
\includegraphics[width=2.5in, trim = 1.1mm 0mm 1mm 2mm, clip]{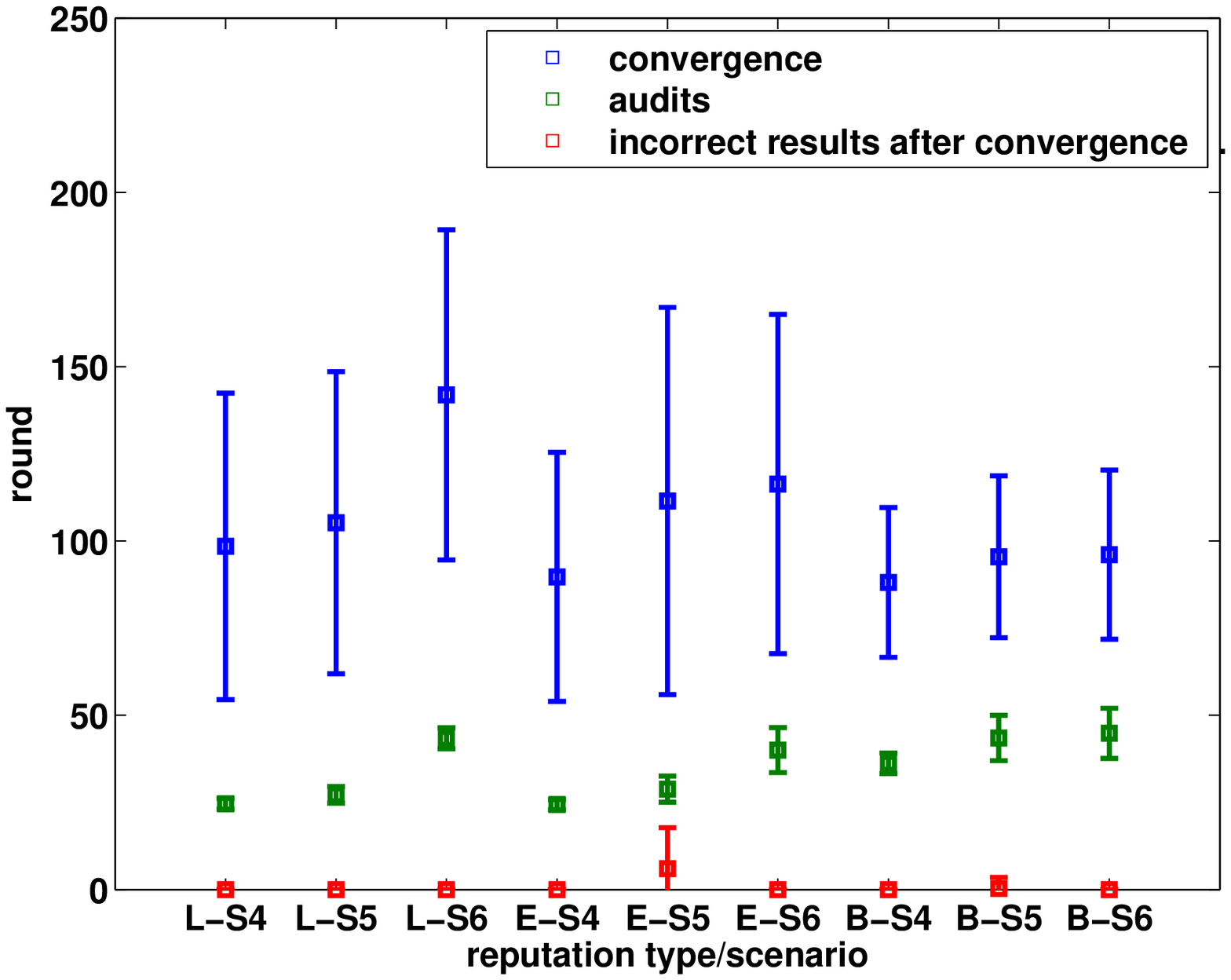}\\
(a2)&(b2)\vspace{-1em}\\
\end{array}$
\caption{\small Simulation results with partial availability: (a1)-(a2) initial $p_\VRF=0.5$, (b1)-(b2) initial $p_\VRF=1$ . For (a1)-(b1) The bottom (red) errorbars present the number of incorrect results accepted until convergence ($p_\VRF=p_\VRF^{min}$). For (a2)-(b2) the bottom (red) errorbars present the number of incorrect results accepted after convergence. For all plots, the middle (green) errorbars present the number of audits until convergence; and finally the upper (blue) errorbars present the number of rounds until convergence, in 100 instantiations. The x-axes symbols are as follows, L: reputation \typeA, E: reputation \typeB , B: reputation~\typeD , S1: 9 altruistic workers with $d=1$, S2: 1 altruistic with $d=1$ and 8 altruistic workers with $d=0.5$, S3: 1 altruistic with $d=1$ and 8 malicious workers with $d=0.5$, S4: 9 rational workers with $d=1$, S5: 1 rational with $d=1$ and 8 rational workers with $d=0.5$, S6: 1 rational with $d=1$ and 8 malicious workers with $d=0.5$.}
\label{unrel-alt-allpa}\vspace{-3em}
\end{figure}

Figure~\ref{unrel-alt-allpa} (a1)-(b1) compares a base case where all workers are altruistic with $d=1$ (scenario S1) with scenarios where 1 altruistic worker exists with $d=1$ and the rest of the workers are either altruistic (scenario S2) or malicious (scenario S3) with a partial availability $d=0.5$. Our base case S1 is the optimal scenario, and the mechanism should have the best performance with respect to metrics (1)-(3); this is confirmed by the simulations as we can observe. For scenario S2, where the 8 altruistic workers have $d=0.5$, reputations \typeA and \typeB are performing as good as the base case. While \typeD is performing slightly worse than the base case. Comparing the different reputation types for scenarios S1 and S2, it is clear that, for all metrics, \typeA and \typeB are performing better than \typeD. Moving on to scenario S3, where 8 malicious workers with $d=0.5$ exist, as expected, the mechanism is performing worse according to our reputation metrics. What is interesting to observe, though, is that reputation \typeD is performing much better than the other two reputation types. It is surprising to observe, for reputation \typeD, how close are the results for scenario S2 and especially scenario S3 to the base case S1. We believe that this is due to the nature of reputation \typeD, which keeps reputation to zero until a reliability threshold is achieved. From the observation of Figure~\ref{unrel-alt-allpa} (a1)-(b1), we can conclude that, if there is information on the existence of malicious workers in the computation, a ``safer'' approach would be the use of reputation \typeD. The impact of $p_\VRF$ on the performance of the mechanism, in the particular scenarios, as it is shown on Figure~\ref{unrel-alt-allpa} (a1)-(b1), in all cases setting $p_\VRF=0.5$ initially improves the performance of the mechanism.

The results of Figure~\ref{unrel-alt-allpa} (a1)-(b1) are confirmed by Theorem~\ref{theorem_t1_t2}. Through the simulation results, we have observed that eventual correctness happens (i.e., no more erroneous results are further accepted) when the system converges, for the depicted scenarios. As for Theorem~\ref{theorem_t4} we have observed that, although the condition of having 5 altruistic with $d=1$ is not the case for scenarios S2 and S3, in the particular scenarios simulated the system was able to reach eventual correctness. Although from the depicted scenarios reputation \typeD seems like is a good approach, theory tells us that it can only be used when we have info on the workers types.

Figure~\ref{unrel-alt-allpa} (a2)-(b2), depicts more scenarios with different workers types ratios, in the presence of rational and malicious workers. Following the same methodology as before, we compare a base case (scenario S4) where all workers are rational with $d=1$, with a scenarios where one rational with $d=1$ exists and the rest are rational (scenario S5) or malicious (scenario S6) with $d=0.5$. We can observe that in the base scenario S4, the mechanism is performing better than in the other two scenarios, for reputation metrics (1),(2) and (4), independently of the reputation type. What we observe is that the most difficult scenario for the mechanism to handle is scenario S5, independently of the reputation type, because, although the system converges, eventual correctness has not been reached and the master is accepting incorrect replies for a few more rounds before reaching eventual correctness. This is due to the ratio of the workers' type, and some rational workers that have not been fully reinforced to a correct behavior may have a greater reputation than the rational worker with $d=1$, while the master has already dropped $p_\VRF=p_\VRF^{min}$. That would mean that the master would accept the result of the majority that might consist of rational workers that cheat. As we can see, \typeB is performing worse than the other two types, based on metric (4). 
As for reputation \typeA we can see that, for scenarios S4 and S5, although the variation on the convergence round is greater than reputation~\typeD, this is traded for half the auditing that reputation \typeD requires. As for scenario S6 (with malicious workers), reputation \typeA converges much slower, while the number of audits is roughly the same, compared to reputation \typeD. This observation gives a slight advantage to reputation \typeD for scenario S6, while reputation \typeA has an advantage on S5.\vspace{-.8em}          

\paragraph{\bf\em Discussion.}
One conclusion that is derived by our simulations is that, in the case of full availability, reputation~\typeD is not a desirable reputation type if the pool of workers is large. As simulations showed us, convergence is slow, and expensive in terms of auditing. One could select one of the other two reputation types (according to the available information on the ratio of workers' type), since accepting a few more incorrect results is traded for fast eventual correctness and low auditing cost. Additionally, in the scenario with full availability we have noticed that, selecting initially $p_\VRF=1$ is a ``safer'' option to have small number of incorrect results accepted, if no information on the system is known and the master is willing to invest a bit more on auditing.

For the case of partial availability, the simulations with only altruistic or with altruistic and malicious converged in all cases. This was expected due to the analysis in all cases except in S2 with reputation \typeD, when we expected to see some rounds after convergence with no replies. The fact is that the altruistic worker with full availability
was able to be selected forever in al cases.
Simulations have also shown that, in the presence of malicious and altruistic workers, reputation \typeD has an advantage compared to the other two types. Finally, it is interesting to observe that, in the partial availability case with only rational workers, our mechanism has not reached eventual correctness when the system has converged, but a few rounds later. This means that, although the rational workers are partially available, the mechanism is able to reinforce them to an honest behavior eventually.\vspace{-1em}

%% file: appendix.tex

\section*{Appendix}

\section{Omitted Proofs}
\label{app:Proofs}

\paragraph{\bf Proof of Observation~\ref{obs1}.}
Let $W^1$ be the subset of $n$ workers chosen by the master in the first round. Since initially there is no knowledge on the type of each worker, there is a positive probability $p$ that all workers in $W^1$ are malicious. By assumption $W^r=W^1$ for all $r>1$, and hence there is a probability $p>0$ that the workers are chosen by the master in each round are all malicious. Assume this happens, then we claim that eventual correctness cannot be satisfied. Assume otherwise; hence, by definition of eventual correctness, there is a round $r_0$ such that in all rounds $r \geq r_0$ the master uses $p_\VRF=p_\VRF^{min}<1$. But then, the probability that the master obtains the correct task result in round $r$ cannot be 1, as required by the eventual correctness property, since with probability $1-p_\VRF>0$ all the received replies are incorrect and the master does audit them. Hence, eventual correctness cannot be satisfied.
For the sake of contradiction, assume that the master does not change workers over rounds and eventual correctness is achieved. That is, there is a round $r_0$ such that for all rounds $r \geq r_0$ the master uses $p_\VRF=p_\VRF^{min}<1$ and obtains the correct answer with probability 1, even though the master never changes workers. Let $W$ be the subset of n workers chosen by the master that will never change. Given that the type of each worker is unknown, that workers are chosen uniformly at random, and that there are at least n malicious workers, there is a probability $p>0$ that the master chooses only malicious workers. Consider round $r_0$. In round $r_0$, there is a probability $1-p_\VRF>0$ that the master does not audit. Thus, the probability that the master obtains the correct answer in round $r_0$ is 
$1-p(1-p_\VRF)<1$, which is a contradiction.
\qed

\paragraph{\bf Proof of Observation~\ref{obs2}.}
For the sake of contradiction assume that every non-malicious worker $i$ has $d_i<1$ and eventual correctness is satisfied. Then, by definition
of eventual correctness, there is a round $r_0$ such that in all rounds $r \geq r_0$ the master uses $p_\VRF=p_\VRF^{min}<1$.
In any round $r \geq r_0$ there is a positive probability that the master does not audit and all the replies received (if any) are incorrect. Then, there is a positive probability that the master does not obtain the correct task result, which is a contradiction.
\qed

\paragraph{\bf Proof of Theorem~\ref{theorem_t1_t2}.}
First, observe that the responsiveness reputation of worker $i$ will always be $\rho_{rs_i}=1$, since $d_i=1$. In fact, all workers $j$ with $d_j=1$ will have responsiveness reputation $\rho_{rs_j}=1$ forever. Moreover, for any worker $k$ with $d_k<1$ that is selected by the master an infinite number of rounds, with probability 1 there is a round $r_k$ in which $k$ is selected but the master does not receive its reply. Hence, $\rho_{rs_k}(r)<1$ for all $r>r_k$.

Let us now consider truthfulness reputation (of types \typeA and \typeB). All altruistic workers $j$ (including $i$) have truthfulness reputation $\rho_{tr_j}=1$ forever, since the replies that the master receives from them are always correct. Malicious workers, on the other hand, fall in one of two cases. A malicious worker $k$ may be selected a finite number of rounds. Then, there is a round $r'_k$ after which it is never selected. If, conversely, malicious worker $k$ is selected an infinite number of rounds, since $d_k>0$ and $p_\VRF \geq p_\VRF^{min}>0$, its replies are audited an infinite number of rounds, and there is a round $r'_k$ so that $\rho_{tr_k}(r)<1/n$ for all $r>r'_k$.

Hence, there is a round $R$ such that, for all rounds $r >R$, (1) every malicious worker $k$ has $\rho_{tr_k}(r)<1/n$ or is never selected by the master, and (2) every worker $k$ with $d_k<1$ has $\rho_{rs_k}(r)<1$ or is never selected by the master. Since there is at least worker $i$ with reputation $\rho_i=1$, we have that among the $n$ workers in $W^r$, for all rounds $r >R$, there is at least one altruistic worker $j$ with $d_j=1$ and $\rho_j=1$, and the aggregate reputation of all malicious workers is less than 1. Hence,
the master always gets correct responses from a weighed majority of workers. This proves the claim.
\qed